\documentclass[12pt,a4paper,reqno]{article}
\usepackage{graphics}
\usepackage{epsfig}

\usepackage{amsmath,amsthm,amssymb}
\newtheorem{theorem}{Theorem}[section]

\newtheorem{proposition}[theorem]{Proposition}
\newtheorem{example}[theorem]{Example}
\theoremstyle{definition}
\newtheorem{definition}[theorem]{Definition}
\theoremstyle{remark}

\numberwithin{equation}{section}

\begin{document}
\title{ Additive Tridiagonal Codes over $\mathbb{F}_{4}$}

\author{N. Annamalai\\
Assistant Professor\\
Indian Institute of Information Technology Kottayam\\
Pala-686635, Kerala, India\\
{Email: algebra.annamalai@gmail.com}
\bigskip\\
Anandhu Mohan\\
Research Scholar\\
Indian Institute of Information Technology Kottayam\\
Pala-686635, Kerala, India\\
{Email: anandhum.phd203002@iiitkottayam.ac.in }
\bigskip\\
C. Durairajan\\
Associate  Professor\\
Department of Mathematics\\ 
School of Mathematical Sciences\\
Bharathidasan University\\
Tiruchirappalli-620024, Tamil Nadu, India\\
{Email: cdurai66@rediffmail.com}
\hfill \\
\hfill \\
\hfill \\
\hfill \\
{\bf Proposed running head:} Additive Tridiagonal Codes over $\mathbb{F}_{4}$}
\date{}
\maketitle

\newpage

\vspace*{0.5cm}
\begin{abstract} 
In this paper, we introduce a additive Tridiagonal and Double-Tridiagonal codes over $\mathbb{F}_4$ and then we study the properties of the code. Also, we find the number of additive Tridiagonal codes over $\mathbb{F}_4.$ Finally, we study the applications of Double-Tridiagonal codes to secret sharing scheme based on matrix projection.
\end{abstract}
\vspace*{0.5cm}

{\it Keywords:} Additive Tridiagonal codes, Double-Tridiagonal Codes, Secret Sharing Scheme.

{\it 2000 Mathematical Subject Classification:}  94B05, 94A62, 94A15.
\vspace{0.5cm}
\vspace{1.5cm}

\section{Introduction}
 Denote the finite field of order 2 by $\mathbb{F}_2$  and the finite field of order 4 by $\mathbb{F}_4 = \{0, 1, w, w^2\},$ where $w^2+ w + 1=0.$
If $\mathbb{F}$ is any field, we define a code of length $n$ over $\mathbb{F}$ to be a subset of the space $\mathbb{F}^n.$  A code is said to be {\it linear} over a field $\mathbb{F}$ if the code is a linear subspace of $\mathbb{F}^n.$ A code $\mathcal{C}$ is said to be an {\it additive code} over $\mathbb{F}$ if it is a subgroup of $\mathbb{F}^n,$ this means that scalar multiples of the codewords do not necessarily belong to the code. We note that for binary codes these two concepts are identical, but for the field of order 4, a code is an additive code over $\mathbb{F}_4$ without being linear. Additive codes over the field $\mathbb{F}_4$ have found numerous applications including being used in $\cite{cal}$ for quantum error-correction and in $\cite{kim}$ for the construction of secret sharing schemes. Algebraic structure of additive conjucyclic codes over $\mathbb{F}_4$ were discussed in \cite{taher}.

An additive code $\mathcal{C}$ over $\mathbb{F}_4$ of length $n$ is an additive subgroup of $\mathbb{F}_4^n.$ The code $\mathcal{C}$ contains $2^k$ codewords for some
$0 \leq k \leq 2n,$ and can be defined by a $k \times n$ generator matrix, with entries from $\mathbb{F}_4,$ whose rows span $\mathcal{C}$ additively. We call $\mathcal{C}$ an $(n, 2^k)$ code.

The Hamming weight of $u \in \mathbb{F}_4^n,$ denoted by $wt(u),$ is the number of nonzero components of $u.$ The Hamming distance between $u$ and $v$ is $wt(u - v).$ The minimum distance of the code $\mathcal{C}$ is the minimal Hamming distance between any two distinct codewords of $\mathcal{C}.$ Since $\mathcal{C}$ is an additive code, the minimum distance is also given by the smallest nonzero weight of any  codeword in $\mathcal{C}.$ An additive code with minimum distance $d$ is called an $(n, 2^k, d)$ code. It follows from the Singleton bound $\cite{pless}$ that any additive $(n, 2^n, d)$ code over $\mathbb{F}_4$ must satisfy
$$d\leq\Big\lfloor \frac{n}{2}\Big\rfloor+1.$$

If a code attains the minimum distance $d$ given by the Singleton bound, it is called an {\it extremal code}. If a code has highest possible minimum distance, but is not extremal, it is called an {\it optimal code}, denoted by $d_{max}.$ If a code has minimum distance $d_{max} - 1,$ it is called {\it near-optimal}.

We say that two additive codes $\mathcal{C}_1$ and $\mathcal{C}_2$ over $\mathbb{F}_4$ are {\it equivalent} provided there is a map sending the codewords of $\mathcal{C}_1$ onto the codewords of $\mathcal{C}_2$ where the map consists of a permutation of coordinates (or columns of the generator matrix), followed by a scaling of coordinates by nonzero elements of $\mathbb{F}_4,$ followed by conjugation of some of the coordinates. The conjugation of $x \in \mathbb{F}_4$ is defined by $\overline{x}= x^2.$ That is, $\overline{0}=0, \overline{1}=1, \overline{w}=1+w$ and $\overline{1+w}=w.$

Recall that for any $x \in\mathbb{F}_4,$ we have that $Tr(x) =x +\overline{x}.$ The trace function is a function from $\mathbb{F}_4$ to $\mathbb{F}_2.$ Let $u =(u_1, u_2,\cdots, u_n), v = (v_1, v_2, \cdots, v_n) \in \mathbb{F}_4^n.$ Then the trace inner-product is defined by
$$\langle u, v\rangle = Tr([u, v])$$
where $[u, v]$ is the standard inner product and
the Hermitian trace inner product of two vectors over $GF(4)$ of length $n$ is given by
$$u * v = Tr(u \cdot \overline{v}) =\sum_{i=1}^n Tr(u_i\overline{v_i}) = \sum_{i=1}^n(u_iv_i^2 + u_i^2v_i) \,(mod 2).$$
We define the dual of the code $\mathcal{C}$ with respect to the Hermitian trace inner product,
$$\mathcal{C}^{\perp} = \{u \in GF(4)^n \mid u * c = 0 \text{ for all } c \in \mathcal{C}\}.$$

The trace dual code with respect to trace inner-product is defined  by
$$\mathcal{C}^{Tr} = \{v \mid \langle v, w \rangle  = 0, \text{ for all } \, w \in C\}.$$

A block code will be called {\it reversible} if the block of digits formed by reversing the order of the digits in a codeword is always another codeword in the same code. That is, if $(c_1,c_2, \cdots, c_n)\in \mathcal{C},$ then $(c_n, c_{n-1}, \cdots, c_1)\in \mathcal{C}.$ 

It is well-known $\cite{cal}$ that additive self-orthogonal codes over $\mathbb{F}_4$ can be used to represent a class of quantum error-correcting codes. Several papers (for example $\cite{cal},$ $\cite{par},$ $\cite{gab},$ $\cite{gul},$ $\cite{hon},$ $\cite{var}$ ) were devoted to classifying or constructing additive self-dual codes over $\mathbb{F}_4.$  

Danielsen and Parker $\cite{dan}$ showed that additive $(n, 2^n)$ codes over $\mathbb{F}_4,$ except for some special cases, have representations as directed graphs. To check whether two
additive codes over $\mathbb{F}_4$ are equivalent, they used a modified version of an algorithm originally devised by Ostergard $\cite{ost}$ for checking equivalence of linear codes. By
using this algorithm, and the fact that codes correspond to directed graphs, they classified additive $(n, 2^n)$ codes over $\mathbb{F}_4$ of length up to 7. Danielsen and Parker
$\cite{dan}$ studied additive circulant codes over $\mathbb{F}_4.$ 

A directed graph is a pair $G = (V, E)$ where $V$ is a set of vertices and $E \subseteq V \times V$ is a set of ordered pairs called edges. A graph with $n$ vertices can be represented by an $n\times n$ adjacency
matrix $\Gamma=(\gamma_{ij})$ where $\gamma_{ij} = 1$ if $(i, j )\in E$  and $\gamma_{ij} = 0$ otherwise.



A {\it directed graph code} is an additive $(n, 2^n)$ code over $\mathbb{F}_4$ that has a generator matrix of the form $A = \Gamma + w I$  where $\Gamma$ is the adjacency matrix of a  directed graph and $I$ is the identity matrix.
\begin{proposition}\cite{dan}
	Given a directed graph code $\mathcal{C}$ with generator matrix $A = \Gamma+w I,$ its dual
	$\mathcal{C}^{\perp}$ is generated by $A^T.$
\end{proposition}
Secret sharing scheme is distributing a secret to a set of participants, in such a way that only certain subsets of them can retrieve the secret. The set of all subsets of participant which are able to retrieve the secret is called the qualified group or access structure of the scheme and those whose are unable is said to be unqualified. This concept was first proposed by G.R Blakley $\cite{blak}$ and A.Shamir $\cite{adi}$, independently in 1979, based on  $(n,m)$ threshold-secret sharing scheme for $n\leq{m}$. In a $(n,m)$ threshold secret sharing scheme, $n$ or more participants can reconstruct the  secret while $(n-1)$ or fewer will not be able to retrieve the secret. The Secret sharing scheme had grown to many branches and ramp secret sharing scheme is a pioneer among them, in which exposed information is proportional to the size of unqualified group.

In the $(n,k,m)$-threshold ramp secret sharing scheme, we can reconstruct the secret from $n$ or more shares, but no information about the secret can be obtained from $n-k$ or fewer shares. Moreover, any $n-l$ shares can recover the secret for $l= 1, 2, \cdots, k-1$. Various research papers were published on ramp secret sharing scheme $\cite{sri,pai,yama,san,mea,bai,mig}$ and ramp secret sharing scheme was able to reduced the size of shares to be distributed. For a matrix $H$, its projection is defined as
$Proj(H)=H(H^{T}H)^{-1}H^{T}$.  The scheme uses matrix projection invariance property to share multiple secrets.

The organization of the paper is as follows:
In section 2, we introduce and study a additive Tridiagonal codes and properties of this  codes over $\mathbb{F}_4.$  In section 3, we study Double-Tridiagonal codes and define a matrix projection of transpose of generator matrix $G$ of a Double-Tridiagonal code. In section 4, we study a $(n,m)$-threshold secret sharing scheme  based on Double-Tridiagonal code with an example.

\section{Additive Tridiagonal Codes over $\mathbb{F}_4$}
In this section, we define and construct additive Tridiagonal codes over $\mathbb{F}_4.$ Also, we find the number of additive Tridiagoal codes over $\mathbb{F}_4.$
\subsection{The Construction of Additive Tridiagonal Codes} 
Danielsen and Parker$\cite{dan}$ introduced the additive circulant codes over $\mathbb{F}_4.$ An additive $(n, 2^n)$ code $\mathcal{C}$ over
$\mathbb{F}_4$ with generator matrix
$$\begin{bmatrix}
	w&a_1&a_2&\cdots&a_{n-1}\\
	a_{n-1}&w&a_1&\cdots&a_{n-2}\\
	a_{n-2}&a_{n-1}&w&\cdots&a_{n-3}\\
	\vdots&\vdots&\vdots&\ddots&\vdots\\
	a_1&a_2&a_3&\cdots&w
\end{bmatrix}
$$
is called an {\it additive circulant code} where $a_i\in \{0, 1\} \subseteq \mathbb{F}_4$ for $1 \leq i \leq n - 1.$ The vector ${\bf a} = (w, a_1, a_2, \cdots, a_{n-1})$ is called a {\it generator vector} for the code $\mathcal{C}.$

Murat SahIn an Haryullah OzImamoglu$\cite{murat}$ generalized the additive circulant codes over $\mathbb{F}_4.$ They define, an additive $(n, 2^n)$ code $\mathcal{C}$ over $\mathbb{F}_4$ with generator matrix
$$\begin{bmatrix}
	w&a_1&a_2&\cdots&a_{n-1}\\
	b_{1}&w&a_1&\cdots&a_{n-2}\\
	b_{2}&b_{1}&w&\cdots&a_{n-3}\\
	\vdots&\vdots&\vdots&\ddots&\vdots\\
	b_{n-1}&b_{n-2}&b_{n-3}&\cdots&w
\end{bmatrix}
$$
is called an {\it additive Toeplitz code} where $a_i, b_i\in \{0, 1\} \subseteq \mathbb{F}_4$ for $1 \leq i \leq n - 1.$ The vector ${\bf a} = (w, a_1, a_2, \cdots, a_{n-1})$ is called an {\it upper generator vector}, the vector ${\bf b}=(w, b_1, b_2, \cdots, b_{n-1})$ is called a {\it lower generator vector} and the ordered pair ${\bf (a, b)}$ is called a {\it generator vector} for the code $\mathcal{C}.$

We define the additive Tridiagonal codes over $\mathbb{F}_4$ as follows:

\begin{definition} \label{def1}
	An additive $(n, 2^n)$ for $n\geq 3$ code $\mathcal{C}$ over $\mathbb{F}_4$ with generator matrix
	$$A=\begin{bmatrix}
		w&a_1&0&\cdots&0&0\\
		b_1&w&a_2&\cdots&0&0\\
		0&b_{2}&w&\cdots&0&0\\
		\vdots&\vdots&\vdots&\ddots&\vdots&\vdots\\
		0&0&0&\cdots&w&a_{n-1}\\
		0&0&0&\cdots&b_{n-1}&w
	\end{bmatrix}_{n\times n}$$
	is called an {\it additive Tridiagonal code} where $a_i, b_i \in \{0, 1\} \subset \mathbb{F}_4$ for $1 \leq i \leq n - 1.$
	
	The vector ${\bf a} = (w, a_1, a_2, \cdots, a_{n-1})$ is called an {\it upper generator vector}, the vector ${\bf b} = (w, b_1, b_2, \cdots, b_{n-1})$ is called a {\it lower generator vector} and the ordered pair
	${\bf (a, b)}$ is called a {\it generator vector} for the code $\mathcal{C}.$
\end{definition}

\begin{example}
	Let the generator matrix of an additive Tridiagonal code $\mathcal{C}$ over $\mathbb{F}_4$ be 
	$$A=\begin{bmatrix}
		w&1&0\\
		0&w&1\\
		0&1&w
	\end{bmatrix}=\Gamma+wI,$$ where $\Gamma=\begin{bmatrix}
		0&1&0\\
		0&0&1\\
		0&1&0
	\end{bmatrix}$ is the adjacency matrix of the directed graph 
	\begin{center}
		\includegraphics{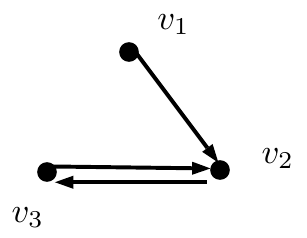}
	\end{center}
	The upper generator vector is ${\bf a} = (w, 0, 1)$ and the lower generator vector is ${\bf b} =(w, 1, 1)$ for the code $\mathcal{C}.$ We get $$\mathcal{C}=\{000, w10, 0w1, 01w, w w^21, w0w, 0w^2w^2, www^2\}.$$
	$\mathcal{C}$ is an additive $(3, 2^3, 2)$ code. Since $d_{max} = 2$ for $n = 3,$ the code $\mathcal{C}$ is optimal.
\end{example}

\begin{example}
	Let the generator matrix of an additive tridiagonal code $\mathcal{C}$ over
	$\mathbb{F}_4$ be
	$$A=\begin{bmatrix}
		w&0&0\\
		1&w&1\\
		0&0&w
	\end{bmatrix}=\Gamma+wI,$$ where $\Gamma=\begin{bmatrix}
		0&0&0\\
		1&0&1\\
		0&0&0
	\end{bmatrix}$ is the adjacency matrix of the directed graph
	\begin{center}
		\includegraphics{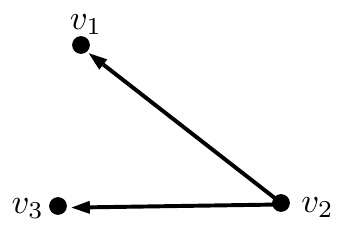}
	\end{center}
	The upper generator vector is $a = (w, 0, 1)$ and the lower generator vector is $b =(w, 1, 0)$ for the code $C.$ We get $$\mathcal{C}=\{000, w00, 1w1, 00w, w^2 w1, w0w, 1ww^2, w^2ww^2\}.$$
	$\mathcal{C}$ is an additive $(3, 2^3, 2)$ code. Since $d_{max} = 1$ for $n = 3,$ the code $\mathcal{C}$ is not optimal.
	
\end{example}
\begin{theorem}
	Let $\mathcal{C}$ be an additive Tridiagonal code over $\mathbb{F}_4$ with generator vector ${\bf v=(a, b)}$ where ${\bf a}=(w, a_1, a_2, \cdots, a_{n-1})$ and ${\bf b}=(w, b_1, b_2, \cdots, b_{n-1}).$ If $a_i=b_{n-i}$ for $i=1, 2, \cdots,n-1,$ then the code $\mathcal{C}$ is a reversible code.
\end{theorem}
\begin{proof}
	Let ${\bf v=(a, b)}$ where ${\bf a}=(w, a_1, a_2, \cdots, a_{n-1})$ and ${\bf b}=(w, b_1, b_2, \cdots, b_{n-1})$  with $a_i=b_{n-i}$ for $i=1, 2, \cdots,n-1$ be a generator vector. Then the generator matrix is
	$$A=\begin{bmatrix}
		w&a_1&0&\cdots&0&0\\
		a_{n-1}&w&a_2&\cdots&0&0\\
		0&a_{n-2}&w&\cdots&0&0\\
		\vdots&\vdots&\vdots&\ddots&\vdots&\vdots\\
		0&0&0&\cdots&w&a_{n-1}\\
		0&0&0&\cdots&a_1&w
	\end{bmatrix}_{n\times n}.$$ Since the reversible of $i$th row is the $(n-i+1)$th row,  the code generated by $\mathcal{C}$ is a reversible code.
\end{proof}

\begin{theorem}
	Let $\mathcal{C}_i$ be an $(n, 2^{k_i}, d_i)$ additive Tridiagonal codes over $\mathbb{F}_4$ for $i=1, 2.$  Then $\mathcal{C}_1\times \mathcal{C}_2$ is a $(2n, 2^{k_1+k_2}, \min\{d_1, d_2\})$ additive Tridiagonal code over $\mathbb{F}_4.$ 
\end{theorem}
\begin{proof}
	Let $A_i$ be a generator matrix for $\mathcal{C}_i$ for $i=1, 2.$ Then the generator matrix of the code $\mathcal{C}_1 \times \mathcal{C}_2$ is $\begin{pmatrix}
		A_1&0\\0&A_2
	\end{pmatrix}.$ Then $\mathcal{C}_1\times \mathcal{C}_2$ is $(2n, 2^{k_1+k_2}, \min\{d_1, d_2\})$ additive tridiagonal code over $\mathbb{F}_4.$
\end{proof}
\begin{definition}
	Let $\mathcal{C}$ be an $(n, 2^n)$ additive Tridiagonal codes over $\mathbb{F}_4.$  Then  the conjugation of $\mathcal{C}$ is defined by $$\overline{\mathcal{C}}=\{\overline{c}\in \mathbb{F}_4^n\mid c\in \mathcal{C}\}$$ where $\overline{c}=(\overline{c_0}, \overline{c_1}, \cdots, \overline{c_{n-1}})$ and $c=(c_0, c_1, \cdots, c_{n-1}).$
\end{definition}
The proof of the following theorem is simple and hence omitted.
\begin{theorem}
	Let $\mathcal{C}$ be an $(n, 2^n)$ additive Tridiagonal codes over $\mathbb{F}_4$ with generator matrix $A$ and generator vector $v=(a, b)$ where $a=(w, a_1, a_2, \cdots, a_{n-1})$ and $b=(w, b_1, b_2, \cdots, b_{n-1}).$ Then the generator matrix of the conjugation code $\overline{\mathcal{C}}$ is $I+A$ where $I$ is the $n\times n$ identity matrix.
\end{theorem}

Since we have two choices for each $a_i,$  $1 \leq i \leq n - 1,$  in the upper generator vector $a =(w, a_1,  \cdots, a_{n-1}),$ the number of upper generator vectors except for the vector $(w, 0, 0, \cdots, 0)$ is $2^{n-1} - 1$. Similarly, the number of lower generator vectors is $2^{n-1} - 1$. We excluded the codes such that their upper or lower generator vectors
are $(w, 0, 0, \cdots, 0)$ since the minimum distances of these codes are $1.$ So, there are $(2^{n-1} - 1)^2$ additive Tridiagonal codes of length $n,$ some of them may be equivalent.

\section{Double-Tridiagonal Codes over $\mathbb{F}_4$}
In this section, we introduce Double-Tridiagonal codes and discuss the applications to secret sharing scheme based on matrix projection.
\begin{definition}
	A linear code $\mathcal{C}$ of length $2n$ is said to be {\it Double-Tridiagonal} if the generator matrix $G$ of $\mathcal{C}$ is of the form $(I\mid A)$ where $I$ is the identity matrix of size $n\times n$ and $A$ is a $n\times n$ generator matrix of a additive Tridiagonal code. 
\end{definition}
\begin{example}
	Let $A=\begin{bmatrix}
		w&1&0\\
		0&w&1\\
		0&1&w
	\end{bmatrix}$ be a generator matrix of an additive Tridiagonal code. Then $$G=\begin{bmatrix}
		1&0&0&|& w&1&0\\
		0&1&0&|&0&w&1\\
		0&0&1&|&0&1&w
	\end{bmatrix}$$ is a generator matrix of a Double-Tridiagonal code $\mathcal{C}$ over $\mathbb{F}_4.$ 
	
\end{example}
We denote $G=(I\mid A)$ is a generator matrix of Double-Tridiagonal matrix of size $n\times 2n$  where $I$ is the identity matrix of size $n\times n$ and $A$ is a generator matrix of a additive Tridiagonal code of size $n\times n.$
\subsection{Matrix Projection of Transpose of Generator Matrix $G$} For the transpose of the generator matrix $G$ which is of the order $2n\times n$ having rank $n,$ we define $\mathbb{S}=H(H^TH)^{-1}H^T$ where $H=G^{T}$ is the transpose of the generator matrix $G$. The $2n\times2n$ matrix $\mathbb{S}$ is the projection matrix of $H$ and denote $\mathbb{S}=Proj(H)$.

Let $x_{i}$'s be linearly independent $n\times 1$ vectors for $1\leq{i}\leq{n}.$ Now compute $v_{i}=Hx_{i}$ for all $1\leq{i}\leq{n}.$ These $2n\times1$ vectors $v_{i}$ can be arranged as a matrix of the form $K=[v_{1}\, v_{2}\, \cdots\, v_{n}].$
\begin{theorem} For a $2n\times n$ matrix $H$ of rank $n$ and a $2n\times n$ matrix $K=[v_{1}\,  v_{2}\, \cdots\, v_{n}]$ where $v_{i}=H x_{i}$  and $x_{i}$'s are  linearly independent $n\times 1$ vectors for $1\leq{i}\leq{n}.$ Then the projection of the matrices $H$ and $K$ are the same. That is, $Proj(H)=Proj(K).$
\end{theorem}
\begin{proof}
	Let $H$ be an $2n \times n$ matrix of rank $n$ and let $x_i, 1\leq i\leq n,$ be any linearly independent vectors in $\mathbb{F}_2^n.$ Define $v_{i}=H x_{i}$ for $1\leq i\leq n.$ Let $K=[v_{1}\,  v_{2}\, \cdots\, v_{n}].$ Then
	\begin{equation}\label{equation 1}
		K=[v_{1}\,  v_{2}\, \cdots\,v_{n}]=H[x_{1}\, x_{2}\, \cdots \, x_{n}]
	\end{equation}
	Let $X=[x_{1}\, x_{2}\, \cdots \,x_{n}]$. Then the $n\times n$ matrix $X$ is a full rank matrix since each $n$ columns are linearly independent.
	From Equation \ref{equation 1}, we get
	$ K=HX.$\\
	Consider, \begin{align*}
		Proj(K)&=K(K^{T}K)^{-1}K^{T}\\
		&=HX((HX)^{T}HX)^{-1}(HX)^{T}\\
		&=HX(X^{T}H^{T}HX)^{-1}(HX)^{T}\\
		&=HXX^{-1}(H^{T}H)^{-1}(X^{T})^{-1}X^{T}H^{T}\\
		&=H(H^{T}H)^{-1}H^{T}\\
		&=Proj(H).
	\end{align*}
	Thus both $H$ and $K$ have the same projection.
\end{proof}
\section{A $(n, m)$-Threshold Secret Sharing Scheme based on Double-Tridiagonal Code}

In this section, we discuss a $(n, m)$-Threshold secret sharing scheme based on matrix projection using transpose of the generator matrix of the Double-Tridiagonal code.\\

Now we shall form a secret sharing scheme based on the Double-Tridiagonal code by using the concept of matrix projection. The construction of shares for a secret matrix $S$ over $\mathbb{F}_{2}$ of order $2n\times 2n$ can be done by considering the transpose $H$ of the generator matrix $G$ of the Double-Tridiagonal code. Choose $m$ random $n\times 1$ vectors $x_{i}$ such that any $n$ are linearly independent. Now calculate the shares $v_{i}=Hx_{i}\,( mod \ 2)$ for each of the $m$ participants $1\leq{i}\leq{m}$. The maximum possibility for $m$ is $(2^{n}-1)(2^{n}-2)\cdots (2^{n}-2^{n-1})$. Then compute the $Proj(H)=\mathbb{S}.$ Let the matrix $R=(S-\mathbb{S})\ (mod\ 2)$. Distribute the $m$ shares say $v_{i}$'s for the $m$ participants and make $R$ to be public.

To reconstruct the secret $S$ with $n$ or more shares $v_{i}$, first construct a matrix $K=[v_{1}\,  v_{2}\, \cdots\, v_{n}]$ using the $n$ shares. Then calculate $Proj(K)$ which is equal to  $\mathbb{S}$. Now compute the secret $S=(\mathbb{S}+R)\ (mod\ 2).$ Having $n-1$ or fewer shares, one will not be able to recreate the secret since the projection of $H$ cannot be defined in this case, as $H^{T}H$  becomes singular. 

\subsection{Example for Secret Sharing Scheme based on Double-Tridiagonal Code}
A $2$-Threshold easy example can be shown with secret matrix $$S=\begin{bmatrix}
	0&1&1&0\\
	1&0&0&1\\
	0&1&1&0\\
	1&1&0&1
\end{bmatrix}.$$
Let the generator matrix for the Double-Tridiagonal code be
$$G=\begin{bmatrix}
	1&0&|&w&1\\
	0&1&|&1&w
\end{bmatrix}.$$
Then its transpose is $$H=\begin{bmatrix}
	1&0\\
	0&1\\
	w&1\\
	1&w
\end{bmatrix}.$$ 
The projection of this matrix $H$ can be easily calculated as $$\mathbb{S}=(H(H^{T}H)^{-1}H^{T})\ (mod\ 2) =\begin{bmatrix}
	w&0&w^{2}&w\\
	0&w&w&w^{2}\\
	w^{2}&w&1+w&0\\
	w&w^{2}&0&1+w
\end{bmatrix}.$$
Now let us compute the matrix \begin{align*}R&=(S-\mathbb{S})\ (mod\ 2)=\begin{bmatrix}
		-w&1&1-w^{2}&-w\\
		1&-w&-w&1-w^{2}\\
		-w^{2}&1-w&-w&0\\
		1-w&1-w^{2}&0&-w
	\end{bmatrix}\\ &=\begin{bmatrix}
		1+w^{2}&1&w&1+w^{2}\\
		1&1+w^{2}&1+w^{2}&w\\
		1+w&w^{2}&1+w^{2}&0\\
		w^{2}&w&0&1+w^{2}
	\end{bmatrix}.\end{align*}
which will be made to be public.
We shall choose two linearly independent $2\times 1$ vectors say $$x_{1}=\begin{bmatrix} 
	1\\
	1
\end{bmatrix} \text{ and } x_{2}=\begin{bmatrix} 
	0\\
	1
\end{bmatrix}$$
Next compute $v_{i}=H x_{i}$ for $i=1,2$.
So $$v_{1}=\begin{bmatrix} 
	1\\
	1\\
	1+w\\
	1+w
\end{bmatrix} \text{ and } v_{2}=\begin{bmatrix} 
	0\\
	1\\
	1\\
	w
\end{bmatrix}.$$ These vectors which were given to the participants can be used to recreate the secret $S$.
By the collaboration of these vectors,  we will get the matrix $$K=\begin{bmatrix} 
	1&0\\
	1&1\\
	1+w&1\\
	1+w&w
\end{bmatrix}$$
which has projection

$$\mathbb{S}=(H(H^{T}H)^{-1}H^{T})\ mod\ 2 =\begin{bmatrix}
	w&0&w^{2}&w\\
	0&w&w&w^{2}\\
	w^{2}&w&1+w&0\\
	w&w^{2}&0&1+w
\end{bmatrix}.$$
The retrieval of secret $S$ can be done as \begin{align*}
	S&=(\mathbb{S}+R)\ mod\ 2\\&=\begin{bmatrix}
		1+w+w^{2}&1&w^{2}+w^{2}&1+w+w^{2}\\
		1&1+w+w^{2}&1+w+w^{2}&w^{2}+w\\
		1+w+w^{2}&w^{2}+w&w^{2}+w&0\\
		w^{2}+w&w^{2}+w&0&w^{2}+w
	\end{bmatrix}\\&=\begin{bmatrix}
		0&1&1&0\\
		1&0&0&1\\
		0&1&1&0\\
		1&1&0&1
	\end{bmatrix}.
\end{align*}

In this example, if we have only one vector $v_{1}$ instead of two linearly independent vectors $v_{1}$ and $v_{2}$, we will not be able to recreate the secret matrix $S$.
\section*{Conclusion}
In this paper, we introduced and studied a additive Tridiagonal and Double-Tridiagonal codes over the finite field $\mathbb{F}_4.$ Also, we counted the number of additive Tridiagonal codes over $\mathbb{F}_4.$ Finally, we stated an application of Double-Tridiagonal codes to secret sharing scheme based on  matrix projection.

\end{document}